\documentclass[12pt]{article}
\usepackage{amssymb}
\usepackage{amsmath}
\usepackage{color}
\usepackage[colorlinks,linkcolor=blue,citecolor=green]{hyperref}
\newtheorem{theorem}{Theorem}

\newenvironment{proof}[1][Proof]{\textbf{#1.} }{\ \rule{0.5em}{0.5em}}

\usepackage{graphicx}
\allowdisplaybreaks[3]
\begin{document}
\title{\Large \textbf{
Hamiltonian aspects of the kinetic equation for soliton gas}}
\author{ Pierandrea Vergallo$^{1, 2}$, \quad Evgeny V.  Ferapontov$^3$}
\date{
\small $^1$ Department of Mathematics 'Federigo Enriques',  \small University of Milano, Via C. Saldini, 50,  Milano,  20133, Italy\\
\small $^2$  Istituto Nazionale di Fisica Nucleare, Sez.\ Milano\\
\small $^3$ Department of Mathematical Sciences, 
\small Loughborough University, Loughborough, \small Leicestershire LE11 3TU, United Kingdom\\
 }

\maketitle

\begin{abstract}
We investigate Hamiltonian aspects of the integro-differential kinetic equation for dense soliton gas which results as a
 thermodynamic limit of the Whitham equations.
Under a delta-functional  ansatz, the kinetic equation  reduces to a non-diagonalisable  system of hydrodynamic type whose matrix consists of several  $2\times 2$ Jordan blocks. 
We demonstrate that the resulting system possesses local Hamiltonian structures of differential-geometric type, for all standard two-soliton interaction kernels (KdV, sinh-Gordon, hard-rod, Lieb-Liniger, DNLS, and separable cases). In the hard-rod case, we show that  the continuum limit of these structures   provides a local multi-Hamiltonian formulation of the full kinetic equation.

\bigskip


\bigskip


\bigskip

%\noindent
%{\bf To do:} 
% \begin{itemize}
%\item (Piero) Double check signs in Table 3;  I have changed the sign of $h_i$ in the separable case, Table 3: please double check!
%\item (Piero) Double check Casimirs in DNLS
%\item (Piero) Casimirs  $n\geq 3$; see remark at the end of section 2.5 (after arXiv)
%\item (Jenya) Check solutions using the general formula for Casimirs in the examples {\it done, all correct, but DNLS case is still missing, even for n=2}.
%\item (Jenya) Check computations at pages 19--20--21.  {\it See my comments in italic in your section Examples, p.21}
%\item General formula (32) for Momenta
%\item Check sign $-K$ in constant curvature case 
%\item Give one single number (21) to formulas (21a)-(21j)
%\item Page 10,  square brackets should be same size
%\item Try sufficiency by hand
%\item Remove unnecessary numbers in formulas of theorem 1
%\end{itemize}

\end{abstract}

\newpage

\tableofcontents

\bigskip

\section{Introduction}
\label{sec:intro}

Let us begin by introducing the integro-differential kinetic equation for  dense soliton gas \cite{El,EK, EKPZ}:
\begin{equation}\label{gas}
\begin{array}{c}
f_t+(sf)_x=0,\\
\ \\
{\displaystyle s(\eta)=S(\eta)+\mathop{\int} G(\mu, \eta)f(\mu)[s(\mu)-s(\eta)]\ d\mu,}
\end{array}
\end{equation}
where $f(\eta)=f(\eta, x, t)$ is the distribution function, $s(\eta)=s(\eta, x, t)$ is the associated transport velocity, and the integration is carried over the support of $f$. Here the variable $\eta$ is a spectral parameter in the Lax pair associated with the dispersive hydrodynamics; the function $S(\eta)$ (free soliton velocity) and the  kernel $G(\mu, \eta)$ (symmetrised phase shift due to pairwise soliton collisions)  are independent of $x$ and $t$. 
The kernel $G(\mu, \eta)$ is assumed to be symmetric: $G(\mu, \eta)=G(\eta, \mu)$. Equation (\ref{gas}) describes the evolution of a dense soliton gas and represents a broad generalisation of Zakharov's kinetic equation for rarefied soliton gas \cite{Z}. It has appeared independently in the context of generalised hydrodynamics of multi-body quantum integrable systems \cite{Do}.
 In the special case 
$$
S(\eta)=4\eta^2, \qquad G(\mu, \eta)=\frac{1}{\eta \mu} \log \bigg|{ \frac{\eta-\mu}{\eta+\mu}}\bigg|,
$$
system (\ref{gas}) was derived in \cite{El} as a thermodynamic limit of the KdV Whitham equations. 

It was demonstrated in \cite{PTE} that under a delta-functional ansatz,
\begin{equation}\label{del}
 f(\eta, x, t)=\sum_{i=1}^{n}u^i(x, t)\ \delta(\eta-\eta^i(x, t)),
\end{equation}
system (\ref{gas}) reduces to a $2n\times 2n$ quasilinear system for $u^i(x, t)$ and $\eta^i(x, t)$,
\begin{equation}\label{uv}
u^i_t=(u^iv^i)_x, \qquad \eta^i_t=v^i\eta^i_x,
\end{equation} 
where $v^i\equiv -s(\eta^i, x, t)$ can be recovered from the linear system
\begin{equation}\label{vlin}
v^i=-S(\eta^i)+\sum_{k\ne i}\epsilon^{ki}u^k(v^k-v^i), \qquad  \epsilon^{ki}={G(\eta^k, \eta^i)}, \ k\ne i.
\end{equation}
Ansatz (\ref{del}) with $\eta^i(x, t)=const$ was discussed  in \cite{EKPZ}. In this case, the last $n$ equations (\ref{uv}) are satisfied identically, while the first $n$ equations constitute an integrable diagonalisable linearly degenerate system whose Hamiltonian aspects (both local and nonlocal) were explored in \cite{DZP, Bul}. The case of non-constant $\eta^i(x, t)$ was investigated recently  in \cite{FP}:  the matrix of the corresponding system (\ref{uv}) is reducible to $n$ Jordan blocks of size $2\times 2$, furthermore, it was shown that the system is  linearly degenerate  and integrable by a suitable extension of the generalised hodograph method of \cite{Tsarev, Tsarev1}. 
Following \cite{PTE}, let us introduce the new  variables $r^i$ by the formula
$$
r^i=-\frac{1}{u^i}\left(1+\sum_{k\ne i}\epsilon^{ki}u^k \right).
$$
In the dependent variables $r^i, \eta^i$, system (\ref{uv}) reduces to block-diagonal form,
\begin{eqnarray}
\label{J1}
&&
\begin{array}{l} 
r^i_t=v^i r^i_x+p^i\eta^i_x,  \\
\eta^i_t=v^i\eta^i_x,  
\end{array}
\end{eqnarray}
$i=1, \dots, n$, which consists of $n$ Jordan blocks of size $2\times 2$. Here the coefficients $v^i$ and $p^i$ can be expressed in terms of $(r, \eta)-$variables as follows. Let us introduce the $n\times n$ matrix $\hat \epsilon$ with diagonal entries $r^1, \dots, r^n$ (so that $\epsilon^{ii}=r^i$) and off-diagonal entries $\epsilon^{ki}={G(\eta^k, \eta^i)}, \ k\ne i$. 
 Note that this matrix is symmetric due to the symmetry of the kernel $G$. 
 Define another symmetric matrix $\hat \beta=-\hat \epsilon^{-1}$. Explicitly, for $n=2$ we have
 $$
 \hat \epsilon=\left(
 \begin{array}{cc}
 r^1&\epsilon^{12}\\
 \epsilon^{12}&r^2
 \end{array}
 \right), \qquad
  \hat \beta=\frac{1}{r^1r^2-(\epsilon^{12})^2}\left(
 \begin{array}{cc}
 -r^2&\epsilon^{12}\\
 \epsilon^{12}&-r^1
 \end{array}
 \right).
 $$
Denote $\beta_{ki}$ the matrix elements of $\hat \beta$ (indices $k$ and $i$ are allowed to coincide). Introducing the notation $\xi^k(\eta^k)=-S(\eta^k)$, we have  the following formulae for $u^i, v^i$ and $p^i$ \cite{PTE}:
\begin{equation}\label{vp}
u^i=\sum_{k=1}^n\beta_{ki}, \ \  v^i=\frac{1}{u^i}\sum_{k=1}^n \beta_{ki}\xi^k, \ \
p^i=\frac{1}{u^i}\left(\sum_{k=1}^n \epsilon^{ki}_{,\eta^i}(v^k-v^i)u^k+(\xi^i)' \right),
\end{equation}
where we use the notation $\epsilon^{ki}_{,\eta^i}$ to indicate partial derivative with respect to  $\eta^i$.

\medskip

\noindent{\bf Remark.} In the simplest case $n=1$, we have $\hat \epsilon =r^1, \ \hat \beta=-\frac{1}{r^1}, \ u^1=-\frac{1}{r^1}, \ v^1=\xi^1, \ p^1=-(\xi^1)'u^1$, so that system (\ref{uv}) assumes Hamiltonian form
$$
\left(\begin{array}{c}
u^1_t\\
\eta^1_t
\end{array}\right)=\left(\begin{array}{cc}
0& 1\\
1&0
\end{array}\right)\frac{d}{dx}
\left(\begin{array}{c}
\partial h/\partial u^1\\
\partial h/\partial \eta^1
\end{array}\right),
$$
with the Hamiltonian density $h=\zeta^1(\eta^1)\, u^1$ where $(\zeta^1)' =\xi^1$ (in fact, this system possesses infinitely many local Hamiltonian structures).
\medskip

In Section \ref{sec:red}, we investigate Hamiltonian aspects of system (\ref{J1}) for $n\geq 2$. Our main result is that for the existence of a local Hamiltonian structure, it is necessary and sufficient that the phase shift function $G(\mu, \eta)$ has the form
\begin{equation}\label{G}
G(\mu, \eta)=\phi(\mu)\phi(\eta)g[a(\mu)-a(\eta)]
\end{equation}
where $\phi$, $g$ and $a$ are some functions of the indicated arguments (for the symmetry of $G$, the function $g$ must be even).
In Sections \ref{sec:1}, \ref{sec:2} we construct local Hamiltonian structures for delta-functional reductions of all standard cases of the soliton gas equations  presented in Table 1 below.

 \begin{center}
 \footnotesize{Table 1. Types of soliton gas equations}
 
 \vspace{2mm}
 
\begin{tabular}{|l|l|l|}
 \hline
kinetic equation & $S(\eta)$  & $G(\mu, \eta)$ \\
 \hline
KdV soliton gas\vphantom{$\frac{\frac{A}A}{\frac{A}A}$} &
$4\eta^2$ & % parametrized by
$\frac{1}{\eta \mu} \log \bigg|{ \frac{\eta-\mu}{\eta+\mu}}\bigg|$\\
 \hline
sinh-Gordon soliton gas\vphantom{$\frac{\frac{A}A}{\frac{A}A}$} &
$\tanh \eta$ &  % parametrized by
$\frac{1}{\cosh \eta \cosh  \mu}\  \frac{a^2\cosh(\eta-\mu)}{4\sinh^2(\eta-\mu)}$\\
\hline
hard-rod gas\vphantom{$\frac{\frac{A}A}{\frac{A}A}$} &
$\eta$ &  % parametrized by
$-a$\\
\hline
Lieb-Liniger gas\vphantom{$\frac{\frac{A}A}{\frac{A}A}$} &
$\eta$ &  % parametrized by
$\frac{2a}{a^2+(\eta -\mu)^2}$\\
\hline
DNLS soliton gas\vphantom{$\frac{\frac{A}A}{\frac{A}A}$} &
$\eta$ &  % parametrized by
$\frac{1}{2\sqrt{\eta^2-1}\sqrt{\mu^2-1}}\log  \left( \frac{(\eta-\mu)^2-\left(\sqrt{\eta^2-1}+\sqrt{\mu^2-1}\right)^2}{(\eta-\mu)^2-\left(\sqrt{\eta^2-1}-\sqrt{\mu^2-1}\right)^2}\right)$  \\
\hline
separable case\vphantom{$\frac{\frac{A}A}{\frac{A}A}$} &
arbitrary &  % parametrized by
$\phi(\eta)+\phi(\mu)$\\
\hline
general case\vphantom{$\frac{\frac{A}A}{\frac{A}A}$} &
arbitrary &  % parametrized by
$\phi(\mu)\phi(\eta)g[a(\mu)-a(\eta)]$\\
\hline
\end{tabular}
%\caption{table}
%\footnotesize{Table 1. Types of soliton gases.}
 \end{center}
Note that all  cases of Table 1 fall into class (\ref{G}). We refer to \cite{El, EK, CER, Do, Do1, Do2, Spohn} for further discussion and references. For $n=2$, a local Hamiltonian formulation of system (\ref{J1}) was established  in \cite{VerFer23} for all cases of Table 1; here we generalise the results of \cite{VerFer23} to  arbitrary $n$. Finally, in Section \ref{sec:flat} we calculate flat coordinates (Casimirs) and momenta for the Hamiltonian structures constructed in Section \ref{sec:1}.

In Section \ref{sec:hr}, we discuss  continuum limit of the multi-Hamiltonian formalism constructed in Section \ref{sec:2} to obtain a local multi-Hamiltonian formulation of the full kinetic equation for the hard-rod gas. This is done by  rewriting the kinetic equation of the hard-rod gas as an infinite (linearisable) hydrodynamic chain.
In this connection, we refer to \cite{Gib} for Hamiltonian aspects of  integro-differential Vlasov equations, as well as to \cite{GibRai} for calculations of various differential-geometric quantities (Haantjes tensor, curvature tensor, etc) in the Vlasov picture. Note that non-local Hamiltonian formalism of the full kinetic equation for dense soliton gas was developed in \cite{Bul}.

\section{Hamiltonian formulation of delta-functional reductions}
\label{sec:red}
System (\ref{J1}) governing delta-functional reductions of the kinetic equation belongs to a general class of quasilinear systems of the form
 \begin{equation}\label{q}
{R}_{t}=A({R}){R}_{x},
\end{equation}
where ${R}=(R^1, ..., R^N)^T$ is the (column) vector of dependent variables and  $A$ is an $N\times N$ matrix. In particular, for system (\ref{J1}) we have $N=2n, \ R=(r^1,\eta^1, \dots, r^n, \eta^n)^T$, and the matrix $A$ is block-diagonal with $n$ Jordan blocks of size $2\times 2$. 
Systems (\ref{q}) whose matrix $A$ has non-trivial Jordan block structure  have appeared  in numerous applications such  as degenerations of hydrodynamic systems associated with multi-dimensional hypergeometric functions  \cite{KK}, in the context of parabolic regularisation of the Riemann equation \cite{KO2}, as  reductions of hydrodynamic chains and linearly degenerate dispersionless PDEs in 3D \cite{Pavlov1},  in the context of Nijenhuis geometry \cite{BKM}, and as primary flows of non-semisimple Frobenius manifolds \cite{LP}.  
It was shown in \cite{XF} that  {\it integrable} systems of  Jordan block type are governed by the modified KP hierarchy. The subject of this paper, delta-functional reductions of the kinetic equation, leads to particularly interesting integrable examples of type (\ref{q}) where the matrix $A$ consists of several $2\times 2$ Jordan blocks, see  \cite{PTE, FP, VerFer23}.

After introducing Hamiltonian formalism of differential-geometric type, we establish Hamiltonian formulation of delta-functional reductions for all standard two-soliton interaction kernels $G(\mu, \eta)$ presented in Table 1.

\subsection{Hamiltonian operators of differential-geometric type, Tsarev's conditions}

We will be looking at systems (\ref{q}) representable in Hamiltonian form,
\begin{equation*}
R^i_t=B^{ij}\frac{\delta H}{\delta R^j},
\end{equation*}
where $B^{ij}$ is a Hamiltonian operator of Dubrovin-Novikov type,
\begin{equation}\label{DN}
B^{ij}=g^{ij}(R)\partial_x+\Gamma^{ij}_k(R)R^k_x,
\end{equation}
and $H=\int h(R)\, dx$ is a Hamiltonian with density $h(R)$.
The conditions for operator (\ref{DN}) to be Hamiltonian  were obtained in \cite{DN83}. In particular, if $\det(g^{ij})\neq 0$, then $g=\{g^{ij}\}$ is a contravariant flat metric and $\Gamma^{ij}_k=-g^{is}\Gamma^{j}_{sk}$ where $\Gamma^j_{sk}$ are Christoffel symbols of the associated Levi-Civita connection. The matrix $A=(A^i_j)$ of the corresponding system (\ref{q}) is given by the formula
$$
A^i_j=\nabla^i\nabla_jh.
$$
Thus, to specify a Hamiltonian structure of  type (\ref{DN}), it is sufficient to provide the corresponding contravariant flat metric $g$.

In what follows, we will also need a nonlocal generalisation of Dubrovin-Novikov operators (\ref{DN}) of the form
\begin{equation}\label{MF}
B^{ij}=g^{ij}(R)\partial_x+\Gamma^{ij}_k(R)R^k_x+KR^i_x\partial_x^{-1}R^j_x,
\end{equation}
$K$=const \cite{FM}. Operator (\ref{MF}) is Hamiltonian  iff $g$ is a contravariant flat metric of constant curvature $K$ and $\Gamma^{ij}_k=-g^{is}\Gamma^{j}_{sk}$ where $\Gamma^j_{sk}$ are Christoffel symbols of the associated Levi-Civita connection. The matrix $A$ of the corresponding system (\ref{q}) is given by the formula
$$
A^i_j=\nabla^i\nabla_jh+Kh\, \delta^i_j.
$$

It was shown by Tsarev   \cite{Tsarev, Tsarev1} that system (\ref{q}) admits Hamiltonian formulation of type (\ref{DN}) (respectively, (\ref{MF})), 
if and only if the following conditions are satisfied:
\begin{equation}
g^{is}A^j_{s}=g^{js}A^i_s,
\label{cond1}
\end{equation}
\begin{equation}
\label{cond2} 
\nabla_iA^j_k=\nabla_kA^j_i,
\end{equation}
where $\nabla$ denotes covariant derivative in the Levi-Civita connection of the flat (respectively, constant curvature $K$) metric $g$. Note that conditions (\ref{cond1}) imply that, if $A$ has block-diagonal form with $n$ upper-triangular $2\times 2$ Toeplitz blocks as in (\ref{J1}),
\[
\left(\begin{array}{cc}
v^i & p^i\\
0&v^i\\
\end{array}\right),
\]
 then the corresponding metric $g$ (with upper indices) will also have block-diagonal form with $n$ upper-triangular $2\times 2$ Hankel blocks:
\[
\left(\begin{array}{cc}
m_i & n_i \\
n_i&0
\end{array}\right).
\]
Conditions (\ref{cond2}) give differential equations for the metric coefficients $m_i, n_i$. To write them down explicitly let us introduce the following quantities expressed in terms of $v^i, p^i$:
\begin{equation*}
\label{c1}
%\begin{array}{c}
b_i=\frac{v^i_{\eta^i}-p^i_{r^i}}{p^i}, \quad
a_{ij}=\frac{v^i_{r^j}}{v^j-v^i},\quad
b_{ij}=\frac{v^i_{\eta^j}-a_{ij}p^j}{v^j-v^i},
\end{equation*}
\begin{equation*}
c_{ij}=\frac{p^i_{r^j}+a_{ij}p^i}{v^j-v^i}, \quad
d_{ij}=\frac{p^i_{\eta^j}+b_{ij}{p^i}-c_{ij}p^j}{v^j-v^i};
\end{equation*}
here and below we assume $i\ne j$. In terms of these quantities  (which have appeared in \cite{FP}, formulae (15-16), as invariants of commuting flows), conditions (\ref{cond2}) simplify to 
\bigskip

\noindent{\bf Equations for $n_i$:}
\begin{equation}\label{ni}
\frac{n_{i, r^i}}{n_i}=-b_i, \quad \frac{n_{i, r^j}}{n_i}=-2a_{ij}, \quad \frac{n_{i, \eta^j}}{n_i}=-2b_{ij}.
\end{equation}

\bigskip

\noindent{\bf Equations for $m_i$:}
\begin{equation}\label{mi}
\left(\frac{m_{i}}{n_i}\right)_{r^j}=-2c_{ij}, \quad \left(\frac{m_{i}}{n_i}\right)_{\eta^j}=-2d_{ij}.
\end{equation}

\medskip

\noindent 
%Note that one only needs quadratures to calculate the metric coefficients $n_i$ and $m_i$. 
Note   that $n_i$ are defined up to arbitrary multiplicative factors, $n_i\to n_i\, s_i(\eta_i)$, while $m_i$ are defined up to transformations of the form $m_i\to m_i+n_i\, g_i(r^i, \eta^i)$ where $s_i(\eta^i)$ and $g_i(r^i, \eta^i)$ are arbitrary functions of the indicated arguments. 
The general solution of equations (\ref{ni}), (\ref{mi}) has the form
\begin{equation}\label{nm}
n_i=\frac{s_i(\eta^i)}{(u^i)^2}, \quad m_i=-\frac{2s_i(\eta^i)}{(u^i)^3}\displaystyle \sum_{j\neq i}^n u^j\epsilon^{ji}_{,\eta^i} +\frac{g_i(r^i, \eta^i)}{(u^i)^2},
\end{equation}
where $u^i$ are defined in (\ref{vp}) and the functions  $s_i(\eta^i)$ and $g_i(r^i, \eta^i)$ are, as yet, arbitrary. Our next goal is to specify arbitrary functions $s_i(\eta^i)$ and $g_i(r^i, \eta^i)$ so that the corresponding metric $g$ is flat (or has constant curvature $K$). After that, calculation of the Hamiltonian density $h$ will require a simple quadrature.

\subsection{Flatness of the metric tensor}

The operator (\ref{DN}) will be Hamiltonian if the metric $g$ is flat, which is equivalent to the vanishing of the Riemann curvature tensor,
\begin{equation}\label{1o}
R^i_{jkl}=\Gamma^i_{jl,k}-\Gamma^i_{jk,l}+\Gamma^i_{ks}\Gamma^s_{jl}-\Gamma^i_{ls}\Gamma^s_{jk},
\end{equation}
where $\Gamma^{i}_{jk}$ are Christoffel symbols of the  Levi-Civita connection of $g$. 

\begin{theorem}\label{main} The metric specified by (\ref{nm}) is flat if and only if the functions $g_i(r^i, \eta^i)$ are quadratic in $r^i$,
$$
g_i(r^i,\eta^i)=\varphi_i(\eta^i)(r^i)^2+\chi_i(\eta^i)r^i+\psi_i(\eta^i),
$$ 
furthermore, the following conditions must be satisfied:
\begin{equation}
\epsilon^{ij}\left(\chi_i+\chi_j\right)=2\left(s_i\epsilon^{ij}_{,\eta^i}+s_j\epsilon^{ij}_{,\eta^j}+\displaystyle \sum_{k\neq i,j}\epsilon^{ik}\epsilon^{jk}\varphi_k\right), \label{eqdr}
\end{equation}
\begin{equation}
\displaystyle \sum_{k\neq i} \varphi_{k}(\epsilon^{ik})^2+\psi_i=0. \label{eqref2}
\end{equation}
\end{theorem}
\begin{proof}
Let us begin by proving the necessity of the conditions of the theorem. We will use coordinates $(R^1, R^2,  \dots, R^{2n-1}, R^{2n})=(r^1,\eta^1, \dots, r^n, \eta^n)$. Direct calculation of the curvature component $R^{r^i}_{r^ir^i\eta^i}, \ i\in \{1, \dots, n\}$, gives
\begin{equation}\label{der}
%R^i_{i\, i\, i+1}=
- \frac{(\det \epsilon)^2g_{i,r^ir^i}-2\det\epsilon A_{i,i}g_{i,r^i}+\displaystyle 2\sum_{k=1}^ng_k(A_{i,k})^2{+}2\sum_{k,l=1}^nA_{i,k}A_{i,l}\left(s_l\epsilon^{lk}_{,l}+s_k\epsilon^{lk}_{,k}\right)}{2\, s_i\, (\det \epsilon)^2}
\end{equation}
where $A_{i, k}$ is the cofactor of the $n\times n$ matrix $\hat \epsilon$,  i.e.,  the determinant of the minor obtained by eliminating the $i$-th row and the $k$-th column of $\hat \epsilon$.
Indeed, for $j=k=i$ and $l=i+1$, formula (\ref{1o}) gives 
\begin{equation}\label{18}
R^{r^i}_{r^ir^i\eta^i}=\Gamma^{r^i}_{r^i\eta^i,r^i}-\Gamma^{r^i}_{r^ir^i,\eta^i}+\Gamma^{r^i}_{r^is}\Gamma^s_{r^i\eta^i}-\Gamma^{r^i}_{\eta^is}\Gamma^s_{r^ir^i}.
\end{equation}
Using the following expressions for Christoffel symbols of metrics $g$ consisting of $2\times 2$  Hankel blocks,

\begin{align}\label{christof}\begin{split}
&\Gamma^{r^i}_{r^ir^i}=g^{r^i\eta^i}g_{r^i\eta^i,r^i},\\
&\Gamma^{r^i}_{r^i\, \eta^i}=\frac{1}{2}g^{r^i\eta^i}g_{\eta^i\eta^i,r^i},\\
&\Gamma^{r^i}_{\eta^i\,\eta^i}=g^{r^ir^i}g_{r^i\eta^i,\eta^i}-\frac{1}{2}g^{r^ir^i}g_{\eta^i\eta^i,r^i},\\
&\Gamma^{\eta^i}_{\eta^i\,\eta^i}=g^{r^i\eta^i}g_{r^i\eta^i,\eta^i}-\frac{1}{2}g^{r^i\eta^i}g_{\eta^i\eta^i,r^i},\\
&\Gamma^{r^i}_{r^j\, \eta^j}=-\frac{1}{2}g^{r^ir^i}g_{r^j\eta^j,r^i}-\frac{1}{2}g^{r^i\eta^i}g_{r^j\eta^j,\eta^i},\\
&\Gamma^{r^i}_{r^j\, r^j}=-\frac{1}{2}g^{r^ir^i}g_{r^jr^j,r^i}-\frac{1}{2}g^{r^i\eta^i}g_{r^jr^j,\eta^i},\\
&\Gamma^{r^i}_{r^i\, r^j}=\frac{1}{2}g^{r^i\eta^i}g_{r^i\eta^i,r^j},\\
&\Gamma^{\eta^i}_{r^j\, \eta^j}=-\frac{1}{2}g^{r^i\eta^i}g_{r^j\eta^j,r^i},\\
&\Gamma^{\eta^i}_{\eta^j\, \eta^j}=-\frac{1}{2}g^{\eta^ir^i}g_{\eta^j\eta^j,r^i},\\
&\Gamma^{\eta^i}_{r^i\, \eta^i}=\Gamma^{\eta^i}_{r^ir^i}=\Gamma^{r^i}_{kl}=\Gamma^{\eta^i}_{kl}=0,
\end{split}\end{align}
where $i\neq j,k,l$,  the expression \eqref{18} simplifies to
\begin{equation*}
R^{r^i}_{r^ir^i\eta^i}=\Gamma^{r^i}_{r^i\eta^i,r^i}-\Gamma^{r^i}_{r^ir^i,\eta^i}+\displaystyle \sum_{s\neq r^i,\eta^i}\Gamma^{r^i}_{r^i\, s}\Gamma^s_{r^i\eta^i}-\Gamma^{r^i}_{\eta^i\, s}\Gamma^s_{r^ir^i}.
\end{equation*}
Let us  first observe that 
\begin{align*}
&g^{r^ir^i}=m_i, \qquad g^{r^i\eta^i}=n_i, \qquad g^{\eta^i\eta^i}=0,\\
&g_{r^ir^i}=0,\qquad g_{r^i\eta^i}=\frac{1}{n_i},\qquad g_{\eta^i\eta^i}=-\frac{m_i}{n_i^2}.
\end{align*}
Substituing this into \eqref{christof} gives
\begin{subequations}
\begin{align*}
&\Gamma^{r^i}_{r^i\, r^k}\Gamma^{r^k}_{r^i\eta^i}+\Gamma^{r^i}_{r^i\eta^k}\Gamma^{\eta^k}_{r^i\eta^i}=-\frac{m_kn_{i,k}^2+2n_kn_{i,k}n_{i,k+1}}{4n_i^3},
\\
&\Gamma^{r^i}_{r^ir^i,\eta^i}=-\frac{n_{i,i\,i+1}n_i-n_{i,i}n_{i,i+1}}{n_i^2},
\\
&\Gamma^{r^i}_{r^i\eta^i,r^i}=\frac{-m_{i,ii}n_i^2+3m_{i,i}n_{i,i}n_i-2m_in_in_{i,ii}-4m_in_{i,i}^2}{2n_i^3},
\end{align*}
\end{subequations}
leading to the formula
\begin{align*}\begin{split}
R^{r^i}_{r^ir^i\eta^i}&=-\frac{1}{4n_i^3}\Big[\displaystyle \sum_{k\neq i} \left(m_k(n_{i,k})^2+2n_k(n_{i,k})(n_{i,k+1})\right)+8m_in_{i,i}^2+4n_i(n_{i,i})(n_{i,i+1})\Big.\\&\hphantom{ciaociooacioaicoaio}
\Big.-6n_i(n_{i,i})(m_{i,i})+2n_i^2(m_{i,ii})-4n_i^2(n_{i,i\, i+1})-4m_in_{i}n_{i,ii}\Big].\end{split}
\end{align*}
Substituting \eqref{nm} in the above expression, equating to zero the numerator of $R^{r^i}_{r^ir^i\eta^i}$  and differentiating it with respect to  $r^i$ we obtain
\begin{equation*}
0=2\det \epsilon A_{i,i}g_{i,r^i,r^i}+2(\det \epsilon)^2g_{i,r^ir^ir^i}-2(A_{i,i})^2g_{i,r^i}+2(A_{i,i})^2g_{i,r^i}-2\det\epsilon A_{i,i}g_{i,r^ir^i}
\end{equation*}
\begin{equation*}
\hphantom{0}=2(\det \epsilon)^2 g_{i,r^ir^ir^i}.
\end{equation*}
Thus, the functions $g_i(r^i, \eta^i)$ must be quadratic in $r^i$, 
\begin{equation*}
g_i(r^i,\eta^i)=\varphi_i(\eta^i)(r^i)^2+\chi_i(\eta^i)r^i+\psi_i(\eta^i),
\end{equation*}
where $\varphi_i,\chi_i,\psi_i$ are arbitrary functions of $\eta^i$. With this form of $g_i(r^i,\eta^i)$, formula (\ref{der}) reduces to
\begin{align*}
&-\frac{1}{2s_i(\det \epsilon)^2}\Big[ 2(\det \epsilon)^2\varphi_i-2\det\epsilon A_{i,i}(2\varphi_i\, r^i+\chi_i)+\big.
\\
&\hphantom{ciaocioaciaoic}\Big.+\displaystyle 2\sum_{k=1}^n(\varphi_k(r^k)^2+\chi_kr^k+\psi_k)(A_{i,k})^2+2\sum_{k,l=1}^nA_{i,k}A_{i,l}\left(s_l\epsilon^{lk}_{,l}+s_k\epsilon^{lk}_{,k}\right)\Big]
\end{align*}
where $A_{i,j}$ is the $(i,j)$-cofactor of $\hat{\epsilon}$. This expression vanishes if and only if its numerator vanishes:  
\begin{align*} \label{exore}&2(\det \epsilon)^2\varphi_i-2\det\epsilon A_{i,i}(2\varphi_i\, r^i+\chi_i)+
\\
&\hphantom{ciaocioa}+\displaystyle 2\sum_{k=1}^n(\varphi_k(r^k)^2+\chi_kr^k+\psi_k)(A_{i,k})^2+2\sum_{k,l=1}^nA_{i,k}A_{i,l}\left(s_l\epsilon^{lk}_{,l}+s_k\epsilon^{lk}_{,k}\right)=0.
\end{align*}
We remark that this expression
%(\ref{exore}) 
is a polynomial of degree $2n$ in the variables $r^1,\dots,  r^n$, furthermore, it has  degree at most two in each variable $r^k$. We also note that  the coefficients at $r^i$ and $(r^i)^2$ vanish identically. Collecting the coefficients at $\prod_{k\neq i}(r^k)^2$, we obtain 
$$
\displaystyle 2\epsilon^{ij}\left(\sum_{k\neq i} \varphi_{k}(\epsilon^{ik})^2+\psi_i\right).
$$
Finally, collecting the coefficients at $(r^j)\prod_{k\neq i,j} (r^k)^2$ we get 
$$
\displaystyle  2 \epsilon^{ij}\left(\chi_i(\eta^i)+\chi_j(\eta^j)\right)-4\left(s_i(\eta^i)\epsilon^{ij}_{,i}+s_j(\eta^j)\epsilon^{ij}_{,j}+\displaystyle \sum_{k\neq i,j}\varphi_k(\eta^k)\epsilon^{ik}\epsilon^{jk}\right)\label{a3}.
$$
Equating these coefficients to zero gives the necessity part of Theorem \ref{main}.  

\medskip

The sufficiency of  conditions of Theorem \ref{main} can be demonstrated as follows.
Differentiating equation (\ref{eqref2}) by $\eta_k, \ k\ne i$, dividing the result by $(\epsilon^{ik})^2$ and then differentiating again by $\eta^i$, we obtain $\varphi_k(\log \epsilon^{ik})_{\eta^i\eta^k}=0$. Thus, the further analysis splits into two different cases:

\medskip
\noindent{\bf Case 1:} $(\log \epsilon^{ik})_{\eta^i\eta^k}\ne 0$, equivalently, $(\log G)_{\mu \eta}\ne 0$. In this case, $\varphi_i=\psi_i=0$ and the conditions of Theorem \ref{main} simplify to $g_i(r^i,\eta^i)=\chi_i(\eta^i)r^i,$ 
\begin{equation*}
\epsilon^{ij}\left(\chi_i+\chi_j\right)=2\left(s_i\epsilon^{ij}_{,\eta^i}+s_j\epsilon^{ij}_{,\eta^j}\right), 
\end{equation*}
which, in particular, leads to formula (\ref{G}) for the interaction kernel $G(\mu, \eta)$. 

\medskip
\noindent{\bf Case 2:} $(\log \epsilon^{ik})_{\eta^i\eta^k}= 0$, equivalently, $(\log G)_{\mu \eta}= 0$. In this case,  the interaction kernel has  (miltiplicatively) separable form, $G(\mu, \eta)=\phi(\mu)\phi(\eta)$, a particularly interesting special case being the hard-rod gas,  $G(\mu, \eta)=-a$.
 \medskip
 
In both cases, the flatness of the corresponding metric can be obtained by direct calculation. 
\end{proof}

\medskip

We will also need the following nonlocal generalisation of Theorem \ref{main}.

\begin{theorem}\label{nonloc} The metric specified by (\ref{nm}) has constant curvature $K$ (that is, {$R^i_{jkl}=K(\delta^i_kg_{jl}-\delta^i_lg_{jk}$})),  if and only if the functions $g_i(r^i, \eta^i)$ are quadratic in $r^i$,
$$
g_i(r^i,\eta^i)=\varphi_i(\eta^i)(r^i)^2+\chi_i(\eta^i)r^i+\psi_i(\eta^i),
$$ 
furthermore, the following conditions must be satisfied:
\begin{equation*}
\epsilon^{ij}\left(\chi_i+\chi_j\right){-}2K=2\left(s_i\epsilon^{ij}_{,\eta^i}+s_j\epsilon^{ij}_{,\eta^j}+\displaystyle \sum_{k\neq i,j}\epsilon^{ik}\epsilon^{jk}\varphi_k\right), \label{eqdrK}
\end{equation*}
\begin{equation*}
\displaystyle \sum_{k\neq i} \varphi_{k}(\epsilon^{ik})^2+\psi_i={-K}. \label{eqref2K}
\end{equation*}
\end{theorem}

\medskip
Local Hamiltonian formulation (\ref{DN}) of Cases 1 and 2 appearing in the proof of Theorem \ref{main} is discussed in Sections \ref{sec:1} and \ref{sec:2} below. In Case 1, we have finitely many local Hamiltonian structures, while  Case 2 is multi-Hamiltonian with infinitely many compatible local Hamiltonian structures parametrised by arbitrary functions.

\subsection{Local Hamiltonian formulation of the case $(\log G)_{\mu \eta}\ne 0$}
\label{sec:1}

In this case, the metric components $n_i, m_i$ are given by formula (\ref{nm}) where $g_i(r^i, \eta^i)=\chi_i(\eta^i)r^i$. The functions $s_i(\eta^i)$ and $\chi_i(\eta^i)$ are to be recovered from the relations
\begin{equation*}\label{5}
2s_i(\eta^i)\epsilon^{ij}_{,\eta^i}+2s_j(\eta^j)\epsilon^{ij}_{,\eta^j}=\epsilon^{ij} \, (\chi_i(\eta^i)+\chi_j(\eta^j)).
\end{equation*} 
Below we specify the functions $s_i(\eta^i)$ and $\chi_i(\eta^i)$ for all standard cases of the soliton gas equations as listed in Table 1 (excluding the hard-rod case which satisfies $(\log G)_{\mu \eta}= 0$ and is discussed in Section \ref{sec:2}). It turns out that the cases $n=2$ (Table 2) and $n\geq 3$ (Table 3) are different: for $n=2$, we have at least two compatible local Hamiltonian structures, whereas for $n\geq 3$ only one of them survives.

 \begin{center}
 \footnotesize{Table 2: Local Hamiltonian structures for $n=2$}
 
 \vspace{4mm}
 
\begin{tabular}{|l|l|l|}
 \hline
kinetic equation & $s_1(\eta^1),\ s_2(\eta^2)$  & $\chi_1(\eta^1),\ \chi_2(\eta^2)$ \\
 \hline
KdV soliton gas\vphantom{$\frac{\frac{A}A}{\frac{A}A}$} &
$s_1=c_1\eta^1$ & % parametrized by
$\chi_1=-2c_1+c_2$\\
&
$s_2=c_1 \eta^2$ & % parametrized by
$ \chi_2=-2c_1-c_2$\\

 \hline
sinh-Gordon soliton gas\vphantom{$\frac{\frac{A}A}{\frac{A}A}$} &
$s_1=c_1$ & % parametrized by
$\chi_1=-2c_1\tanh \eta^1+c_2$\\
&
$s_2=c_1 $ & % parametrized by
$ \chi_2=-2c_1\tanh \eta^2-c_2$\\

\hline

Lieb-Liniger gas\vphantom{$\frac{\frac{A}A}{\frac{A}A}$} &
$s_1=c_1$ & % parametrized by
$\chi_1=c_2$\\
&
$s_2=c_1$ & % parametrized by
$ \chi_2=-c_2$\\
\hline

DNLS soliton gas\vphantom{$\frac{\frac{A}A}{\frac{A}A}$} &
$s_1=c_1(1-(\eta^1)^2)$ & % parametrized by
$\chi_1=2c_1\eta^1+c_2$\\
&
$s_2=c_1 (1-(\eta^2)^2)$ & % parametrized by
$ \chi_2=2c_1\eta^2-c_2$\\

\hline
separable case\vphantom{$\frac{\frac{A}A}{\frac{A}A}$} 
& $s_1=\frac{c_2\phi^2(\eta^1)+2c_1\phi(\eta^1)+c_4}{2\phi'(\eta^1)} $ &  % parametrized by
$\chi_1=c_1+c_3+c_2\phi(\eta^1)$ \\
& $s_2= \frac{-c_2\phi^2(\eta^2)+2c_1\phi(\eta^2)-c_4}{2\phi'(\eta^2)} $ &  % parametrized by
 $\chi_2=c_1-c_3-c_2\phi(\eta^2)$ \\
\hline

general case\vphantom{$\frac{\frac{A}A}{\frac{A}A}$} &
$s_1=\frac{c_1}{a'(\eta^1)}$ & % parametrized by
$\chi_1=\frac{2c_1\phi'(\eta^1)}{a'(\eta^1)\phi(\eta^1)}+c_2$\\
&
$s_2=\frac{c_1}{a'(\eta^2)}$ & % parametrized by
$ \chi_2=\frac{2c_1\phi'(\eta^2)}{a'(\eta^2)\phi(\eta^2)}-c_2$\\

 \hline

\end{tabular}
%\caption{table}
%\footnotesize{Table 1. Types of soliton gases.}
 \end{center}
In Table 2, $c_i$ are arbitrary constants (responsible for the number of different Hamiltonian structures). The corresponding Hamiltonian densities are of the form
$$
h=u^1 h_1(\eta^1)+u^2h_2(\eta^2)
$$
where the variables $u^i$ are given by formula (\ref{vp}) and the explicit form of the functions $h_i(\eta^i)$ can be found in \cite{VerFer23}.

For $n\geq 3$, only one local Hamiltonian structure survives, namely, the one that corresponds to the constant $c_1$ (in what follows, we set $c_1=1$). 
The corresponding Hamiltonian densities are of the form $h=\sum_{i=1}^n u^i h_i(\eta^i)$ where the functions $h_i$ are given by
\begin{equation}\label{hd}
h_i(\eta^i)=\displaystyle e^{\int{\frac{\chi_i(\eta^i)}{2s_i(\eta^i)}\, d\eta^i}}\int{\frac{\xi^i(\eta^i)e^{-\int{\frac{\chi_i(\eta^i)}{2s_i(\eta^i)}}\, d\eta^i}}{s_i(\eta^i)}\, d\eta^i}.
\end{equation}
Similarly, the densities of momenta are of the form $g=\sum_{i=1}^n u^i g_i(\eta^i)$ where the functions $g_i$ are given by
\begin{equation}\label{gd}
g_i(\eta^i)=\displaystyle e^{\int{\frac{\chi_i(\eta^i)}{2s_i(\eta^i)}\, d\eta^i}}\int{\frac{e^{-\int{\frac{\chi_i(\eta^i)}{2s_i(\eta^i)}}\, d\eta^i}}{s_i(\eta^i)}\, d\eta^i}.
\end{equation}
The results are summarised in Table 3 below where the last two columns give the functions $h_i(\eta^i)$ and $g_i(\eta^i)$.

 \begin{center}
 \footnotesize{Table 3: Local Hamiltonian structures for $n\geq 3$}
 
 \vspace{4mm}
 
\begin{tabular}{|l|l|l|l|l|}
 \hline
kinetic equation & $s_i(\eta^i)$  & $\chi_i(\eta^i)$ &$h_i(\eta^i)$&$g_i(\eta^i)$\\
 \hline
KdV soliton gas\vphantom{$\frac{\frac{A}A}{\frac{A}A}$} &
$\eta^i$ & % parametrized by
$-2$ & $-\frac{4}{3}(\eta^i)^2$&$1$\\

 \hline
sinh-Gordon soliton gas\vphantom{$\frac{\frac{A}A}{\frac{A}A}$} &
$1$ & % parametrized by
$-2\tanh \eta^i$ & $-1$&{$\tanh{\eta^i}$}\\

\hline

Lieb-Liniger gas\vphantom{$\frac{\frac{A}A}{\frac{A}A}$} &
$1$ & % parametrized by
$0$ & $-\frac{1}{2}(\eta^i)^2$&$\eta^i$\\

\hline

DNLS soliton gas\vphantom{$\frac{\frac{A}A}{\frac{A}A}$} &
$1-(\eta^i)^2$ & % parametrized by
$2\eta^i$ & $1$&$\frac{\sin^{-1}{\eta^i}}{\sqrt{1-(\eta^i)^2}}$\\

\hline
separable case\vphantom{$\frac{\frac{A}A}{\frac{A}A}$} 
& $\frac{\phi(\eta^i)}{\phi'(\eta^i)} $ &  % parametrized by
$1$ & $-\sqrt{\phi(\eta^i)}\int^{\eta^i}\frac{\phi'(\eta)S(\eta)}{\phi(\eta)^{3/2}}\ d\eta $&$-2$\\

\hline

general case\vphantom{$\frac{\frac{A}A}{\frac{A}A}$} &
$\frac{1}{a'(\eta^i)}$ & % parametrized by
$\frac{2\phi'(\eta^i)}{a'(\eta^i)\phi(\eta^i)}$ & $-\phi(\eta^i)\int^{\eta^i}{\frac{S(\eta)a'(\eta)}{\phi(\eta)}\, d\eta}$&{$\phi(\eta^i)\int^{\eta^i}{\frac{a'(\eta)}{\phi(\eta)}\, d\eta}$}\\

 \hline

\end{tabular}
 \end{center}

Although systems from Table 3 possess only one local Hamiltonian structure, all of them are multi-Hamiltonian, with compatible (nonlocal) Hamiltonian structures
of type (\ref{MF}), or of a more general nonlocal type studied in \cite{F}. Thus, delta-functional reduction of the KdV soliton gas possesses Hamiltonian operator (\ref{MF}) associated with the constant curvature {$K$} metric (\ref{nm}) with 
$$
s_i(\eta^i)={-}\frac{K}{2}(\eta^i)^3,  \qquad g_i(r^i, \eta^i)=K(\eta^i)^2r^i{-}K;
$$
the corresponding Hamiltonian density is $h=\frac{8}{K} \sum_{i=1}^n u^i$.

\subsection{Local Hamiltonian formulation of the case $(\log G)_{\mu \eta}= 0$}
\label{sec:2}

With $\epsilon^{ij}(\eta^i,\eta^j)=\phi_i(\eta^i)\phi_j(\eta^j)$, equation (\ref{eqref2}) simplifies to   
\begin{equation*}
\sum_{k\neq i} \varphi_{k}(\phi_k)^2+\frac{\psi_i}{(\phi_i)^2}=0.
\end{equation*}
Separating the variables we obtain
\begin{equation*}
\varphi_i=\frac{c_i}{\phi_i^2}, \quad \psi_i={-}\displaystyle \sum_{k\neq i}c_k\, \phi_i^2,
\end{equation*}
where $c_i$ are arbitrary constants. With this form of $\varphi_i$ and $\psi_i$, equation (\ref{eqdr}) simplifies to
\begin{equation}\label{chi}
\left(\chi_i-2s_i\frac{\phi_i'}{\phi_i}\right)+\left(\chi_j-2s_j\frac{\phi_j'}{\phi_j}\right)=2\displaystyle \sum_{k\neq i,j}c_k.
\end{equation}
For $n\geq 3$,  relations  (\ref{chi}) imply
\begin{equation*}
\chi_i=2\frac{s_i\phi_i'}{\phi_i}+\left(\displaystyle \sum_{k\neq i}c_k-c_i\right),
\end{equation*}
ultimately,
\begin{equation*}\label{gi}
g_i(r^i,\eta^i)=\frac{c_i}{\phi_i^2}(r^i)^2+\left(2\frac{s_i(\eta^i)\phi'_i}{\phi_i}+ \displaystyle \sum_{k\neq i}c_k - c_i \right)r^i-\displaystyle \sum_{k\neq i}c_k\phi_i^2.
\end{equation*}
Note that in the special case $n=2$, we have only one relation (\ref{chi}) which gives
$$
\chi_1=2\left(\frac{s_1\phi_1'}{\phi_1}-c\right), \quad \chi_2=2\left(\frac{s_2\phi_2'}{\phi_2}+c\right),
$$
where $c$ is yet another arbitrary constant. The corresponding metrics  are all proved to be flat, thus, we have an infinity of local Hamiltonian structures parametrised by $n$ arbitrary functions $s_i(\eta^i)$ and $n$ arbitrary constants $c_i$ (with an extra constant $c$ appearing in the special case $n=2$). Below we present some explicit formulae for $n=2$.

\bigskip

\noindent{\bf General  case, $n=2$}: $\epsilon(\eta^1, \eta^2)=\phi_1(\eta^1)\phi_2(\eta^2)$. In this case
 the functions $g_1, g_2$ specialise as follows:
\begin{equation*}\label{gg}
\begin{array}{c}
\displaystyle g_1=\,{\frac {{c_1}\,{{(r^1)}}^{2}}{\phi_1^{2}}}+2\left(\frac{\phi_1'}{\phi_1}s_1-c \right){r^1} -{c_2\phi_1^{2}}, \\
\ \\
\displaystyle g_2=\,{\frac {{c_2}\,{{(r^2)}}^{2}}{\phi_2^{2}}}+2\left(\frac{\phi_2'}{\phi_2}s_2+c \right){r^2} -{c_1\phi_2^{2}};
\end{array}
\end{equation*} 
here $s_1(\eta^1),s_2(\eta^2)$ are arbitrary functions  and $c_1,c_2,c$ are arbitrary constants.
For the  particular choice $c_1=c_2=0, \ c=1$, $s_1=\phi_1/\phi_1', \ s_2=-\phi_2/\phi_2'$, we have $g_1=g_2=0$ and the coefficients of the corresponding contravariant metric  simplify to
\begin{equation*}
n_1=  \frac{\phi_1}{\phi_1'}\frac {\left({r^1r^2}-\epsilon^2 \right)^{2}}{ \left( r^2-\epsilon \right) ^{2}},\quad m_1= -2\frac{\epsilon({r^1}-\epsilon)}{(r^2-\epsilon) ^{3}}({r^1r^2}-\epsilon^2)^{2},
\end{equation*}
\begin{equation*}
n_2=-\frac{\phi_2}{\phi_2'}\frac {\left({r^1r^2}-\epsilon^2 \right)^{2}}{ \left( r^1-\epsilon \right) ^{2}}, \quad m_2=  2\frac{\epsilon({r^2}-\epsilon)}{(r^1-\epsilon) ^{3}}({r^1r^2}-\epsilon^2)^{2}.
\end{equation*}
The corresponding Hamiltonian density is given by the formula
\begin{equation*}
h=\frac{(r^1-\epsilon)\displaystyle \int{\frac{\phi_2'\,  \xi^2}{\phi_2}	\, d\eta^2}-(r^2-\epsilon)\displaystyle \int{\,\frac{\phi_1'\, \xi^1}{\varphi_1} d\eta^1}}{r^1r^2-\epsilon^2}.
\end{equation*}

\bigskip

\noindent {\bf Hard-rod gas, $n=2$}: $\epsilon(\eta^1, \eta^2)=-a=const$. In this case, the  functions $g_1, g_2$  can be scaled as follows:
$$
g_1=c_1\, (r^1)^2- 2cr^1- c_2\, a^2, \qquad
g_2=c_2\, (r^2)^2+2cr^2-c_1\, a^2; 
$$
here $s_1(\eta^1),\, s_2(\eta^2)$ are arbitrary functions and $c_1,c_2,c$ are arbitrary constants. 
The coefficients of the corresponding contravariant metric  take the form
\begin{equation*}
n_1=\frac{s_1\, (a^2-r^1r^2)^2}{(a+r^2)^2}, \quad m_1=\frac{\left(c_1\, (r^1)^2-2cr^1-c_2\, a^2\right)\, (a^2-r^1r^2)^2}{(a+r^2)^2},
\end{equation*}
\begin{equation*}
n_2=\frac{s_2\, (a^2-r^1r^2)^2}{(a+r^1)^2}, \quad m_2=\frac{\left(c_2\, (r^2)^2+2cr^2-c_1\, a^2\right)\, (a^2-r^1r^2)^2}{(a+r^1)^2}.
\end{equation*}
Based on the general form of  conservation laws from  \cite{FP}, we obtain the corresponding Hamiltonian density:
\begin{equation*}
h=\frac{(a+r^2)\psi^1+(a+r^1)\psi^2}{r^1r^2-a^2}+\sigma^1(\eta^1)+\sigma^2(\eta^2)
\end{equation*}
where 
$$
\psi^1=c_2k\left(ac_1s_2^2\sigma_2''-cs_1^2\sigma_1''-s_1\left(s_1'c-c_1c_2a^2-{c}^2\right)\sigma_1'+c_1\left(as_2s_2'\sigma_2'+a\xi^2c_2-\xi^1c\right)\right),
$$
$$
\psi^2=c_1k\left(ac_2s_1^2\sigma''_1+cs_2^2\sigma_2''+s_2\left(s_2'c+c_1c_2a^2+c^2\right)\sigma'_2+c_2\left(as_1s_1'\sigma_1'+a\xi^1c_1+\xi^2c\right)\right);
$$
here 
 $$
k=\frac{1}{c_1c_2(c_1c_2a^2+c^2)},
$$
and   $\sigma_1(\eta^1), \sigma_2(\eta^2)$ are solutions of the following ODEs:
\begin{equation*}
(s_1^2)\, \sigma_1'''+(3s_1's_1)\, \sigma_1''+\left(-c_1c_2a^2-c^2+(s_1')^2+s_1s_1''\right)\, \sigma_1'+c_1(\xi^1)'=0,
\end{equation*}
\begin{equation*}
(s_2^2)\, \sigma_2'''+(3s_2's_2)\, \sigma_2''+\left(-c_1c_2a^2-c^2+(s_2')^2+s_2s_2''\right)\, \sigma_2'+c_2(\xi^2)'=0.
\end{equation*}
\medskip
In the case $n\geq 3$ we have
\begin{equation}\label{nmhr}
\begin{array}{c}
n_i=\frac{s_i(\eta^i)}{\prod_{k\ne i}(a+r^k)^2}(\det{\hat{\epsilon}})^2=\frac{s_i(\eta^i)}{(u^i)^2},\\
\ \\
m_i=\frac{\left({c}_ir^i-a\sum_{k\ne i}c_k\right)(a+r^i)}{\prod_{k\ne i}(a+r^k)^2}(\det{\hat{\epsilon}})^2=\frac{\left({c}_ir^i-a\sum_{k\ne i}c_k\right)(a+r^i)}{(u^i)^2},
\end{array}
\end{equation}
where $s_i(\eta^i)$ are arbitrary functions and $c_i$ are arbitrary constants. The special choice $s_i(\eta^i)=1, \ c_i=0$ leads to the flat metric with $n_i=\frac{1}{(u^i)^2}, \ m_i=0$, the corresponding Hamiltonian density is $h=-\frac{1}{2}\sum u^i (\eta^i)^2$. 
We refer to Section \ref{sec:hr} for the explicit form of the limit of this Hamiltonian structure  obtained as $n\to \infty$. 

\medskip

Note that the presence of infinitely many  compatible local Hamiltonian structures is the sign of linearisability of the corresponding system. In Section \ref{sec:hr} we will see that this is indeed the case: the kinetic equation for hard-rod gas is linearisable.
%Note that although this formula gives all Hamiltonian structures for the $n> 2$ hard rod gas, it gives only a subfamily thereof for $n=2$, namely those for which $2c_3+a(c_1-c_2)=0$. 

\subsection{Casimirs and momenta}
\label{sec:flat}

As demonstrated in \cite{FP}, all conserved densities of system (\ref{uv}) are given by the formula
\begin{equation*}\label{cn1}
\sum_{i=1}^nu^i\psi^i(\eta)+\sum_{i=1}^n\sigma^i(\eta^i)
\end{equation*}
where $\sigma^i(\eta^i)$ are arbitrary functions of their arguments and the functions $\psi^i(\eta^1, \dots, \eta^n)$ satisfy the equations
$\psi^i_{,\eta^j}=(\sigma^j)'\epsilon^{ij}, \ j\ne i$. The general conserved density depends on $2n$ arbitrary functions of one variable: $n$ functions $\sigma^i(\eta^i)$, plus extra $n$  functions coming from  $\psi^i$.

In this section we calculate flat coordinates (Casimirs) and momenta for Hamiltonian structures from Table 3. Our first observation is that the first $n$ Casimirs (out of the total number of $2n$) are given by the simple formula
\begin{equation}\label{for}
u^i\psi^i(\eta^i) \quad {\rm where} \quad \psi^i(\eta^i)=\displaystyle e^{\int{\frac{\chi_i(\eta^i)}{2s_i(\eta^i)}\, d\eta^i}},
\end{equation}
no summation, where the functions $s_i$ and $\chi_i$ are the same as in Table 3. The remaining $n$ Casimirs, which are more complicated, are listed below on a case-by-case basis. We will restrict to the case $n=2$. 

\vspace{3mm}

\noindent \textbf{KdV case},   $n=2$. The first two Casimirs, as given by (\ref{for}), have the form
%\begin{equation*} C=\frac{(r^2-\epsilon^{12})\psi^1+(r^1-\epsilon^{12})\psi^2}{r^1r^2-(\epsilon^{12})^2}+\sigma^1(\eta^1)+\sigma^2(\eta^2)\end{equation*}
$$
\frac{u^1}{\eta^1}, \quad \frac{u^2}{\eta^2}.
$$
The next two Casimirs are more complicated:
\begin{align*}
u^1\ \frac{\ln s \ln (s+1)+ \text{dilog}\, s+ \text{dilog}\, (s+1)}{\eta^1}-\eta^2,\\
u^2\ \frac{\ln s \ln (s+1)+ \text{dilog}\, s+ \text{dilog}\, (s+1)}{\eta^2}+\eta^1,
\end{align*}
here $s=\frac{\eta^1-\eta^2}{\eta^1 +\eta^2}$ and $\text{dilog}$ is the dilogarithm function, 
\begin{equation*}
\text{dilog}(x)=\text{Li}_2(1-x)=\displaystyle \int_1^x{\frac{\log t}{1-t}dt}.
\end{equation*}
\medskip
\noindent The density of momentum, as given by formula (\ref{gd}), has the form  
$$
u^1+u^2.
$$

\vspace{3mm}

\noindent \textbf{Sinh-Gordon case}, $n=2$. The first two Casimirs have the form
$$
\frac{u^1}{\cosh{\eta^1}}, \quad \frac{u^2}{\cosh{\eta^2}}.
$$
The next two Casimirs are more complicated:
$$
u^1\frac{a^2}{4\sinh{(\eta^1-\eta^2)}\cosh{\eta^1}}+\sinh{\eta^2}, \quad u^2\frac{a^2}{4\sinh{(\eta^2-\eta^1)}\cosh\eta^2}+\sinh{\eta^1}.
$$
\medskip
\noindent The density of momentum, as given by formula (\ref{gd}), has the form  $$u^1\tanh\eta^1+u^2\tanh\eta^2.$$

\vspace{3mm}

\noindent \textbf{Lieb-Liniger case},   $n=2$.  The first two Casimirs have the form
$$
u^1, \quad u^2.
$$
The next two Casimirs are more complicated:
\begin{align*}
2u^1 \tan^{-1}{\left(\frac{\eta^2-\eta^1}{a}\right)}+\eta^2, \quad 
2u^2 \tan^{-1}{\left(\frac{\eta^1-\eta^2}{a}\right)}+\eta^1.
\end{align*}
\medskip
\noindent The density of momentum, as given by formula (\ref{gd}), has the form  $$u^1\eta^1+u^2\eta^2.$$

\vspace{3mm}

\noindent \textbf{DNLS case}, $n=2$. The first two Casimirs have the form
$$ \frac{u^1}{\sqrt{1-(\eta^1)^2}}, \quad \frac{u^2}{\sqrt{1-(\eta^2)^2}}. $$
The next two Casimirs are more complicated:
\newpage 
\begin{align*}
&\frac{u^1}{\sqrt{(\eta^1)^2-1}}\Bigg[\text{cosh}^{-1}{\eta^2}+2\ln{(e^{x_1}-e^{-x_2})}\left(\text{tanh}^{-1}e^{x_1}-\text{tanh}^{-1}e^{-x_2}\right)\big.\\
&\hphantom{cioacioa}+2\ln{(e^{x_1}-e^{x_2})}\left(\text{tanh}^{-1}e^{x_1}-\text{tanh}^{-1}e^{x_2}\right)-\text{dilog}\left(\frac{e^{x_1}+1}{e^{x_2}+1}\right)\\
&\hphantom{cioaicoaciao}+\text{dilog}\left(\frac{e^{x_1}-1}{e^{x_2}-1}\right)+\text{dilog}\left(\frac{e^{x_1}+1}{e^{-x_2}+1}\right)-\text{dilog}\left(\frac{e^{x_1}-1}{e^{-x_2}-1}\right)\Bigg]-\text{cosh}^{-1}\eta^2,
\\
&\frac{u^2}{\sqrt{(\eta^2)^2-1}}\Bigg[\text{cosh}^{-1}{\eta^1}+2\ln{(e^{x_2}-e^{-x_1})}\left(\text{tanh}^{-1}e^{x_2}-\text{tanh}^{-1}e^{-x_1}\right)\big.\\
&\hphantom{cioacioa}+2\ln{(e^{x_2}-e^{x_1})}\left(\text{tanh}^{-1}e^{x_2}-\text{tanh}^{-1}e^{x_1}\right)-\text{dilog}\left(\frac{e^{x_2}+1}{e^{x_1}+1}\right)\\
&\hphantom{cioaicoaciao}+\text{dilog}\left(\frac{e^{x_2}-1}{e^{x_1}-1}\right)+\text{dilog}\left(\frac{e^{x_2}+1}{e^{-x_1}+1}\right)-\text{dilog}\left(\frac{e^{x_2}-1}{e^{-x_1}-1}\right)\Bigg]-\text{cosh}^{-1}\eta^1,
\end{align*}
where $x_i=\cosh\, \eta^i$.

\medskip
\noindent The density of momentum, as given by formula (\ref{gd}), has the form  
$$u^1\frac{\sin^{-1}{\eta^1}}{\sqrt{1-(\eta^1)^2}}+u^2\frac{\sin^{-1}{\eta^2}}{\sqrt{1-(\eta^2)^2}}.$$

\vspace{3mm}

\noindent \textbf{Separable case},   $n=2$. The first two Casimirs have the form
$$
{u^1}\sqrt{\phi(\eta^1)}, \quad {u^2}\sqrt{\phi(\eta^2)}.
$$
The next two Casimirs are more complicated:
\begin{align*}
u^1\, \frac{\phi(\eta^1)-\phi(\eta^2)}{\sqrt{\phi(\eta^2)}}+\frac{1}{\sqrt{\phi(\eta^2)}}, \qquad
u^2\, \frac{\phi(\eta^2)-\phi(\eta^1)}{\sqrt{\phi(\eta^1)}}+\frac{1}{\sqrt{\phi(\eta^1)}}.
\end{align*}
\medskip
\noindent The density of momentum, as given by formula (\ref{gd}), has the form  $$-2u^1-2u^2.$$

\vspace{3mm}

\noindent \textbf{General case}, $n=2$. The first two Casimirs have the form
$$
{u^1}{\phi(\eta^1)}, \quad {u^2}{\phi(\eta^2)}.
$$
The next two Casimirs are more complicated:
 $$
 u^1\phi(\eta^1){\cal G}[a(\eta^2)-a(\eta^1)]+\int{\frac{a'(\eta^2)}{\phi(\eta^2)}\, d\eta^2}, \qquad u^2\phi(\eta^2){\cal G}[a(\eta^1)-a(\eta^2)]+\int{\frac{a'(\eta^1)}{\phi(\eta^1)}\, d\eta^1},
 $$
where $\cal G$ is the antiderivative of $g$. 
\medskip

\noindent The density of momentum, as given by formula (\ref{gd}), has the form 
$$
u^1\phi(\eta^1)\int^{\eta^1}{\frac{a'(\eta)}{\phi(\eta)}\, d\eta}+u^2 \phi(\eta^2)\int^{\eta^2}{\frac{a'(\eta)}{\phi(\eta)}\, d\eta}.
$$
\medskip

\noindent{\bf Remark.} For all examples above, the last two Casimirs are of the form 
$$u^1\psi^1(\eta^1,\eta^2)+\sigma^2, \qquad u^2\psi^2(\eta^1,\eta^2)+\sigma^1,$$
where 
\begin{align*}
\psi^1_{,\eta^1}=\frac{-2\sigma_2'\, \epsilon\, s_2 + \psi^1\chi_1}{2s_1}, \qquad \psi^1_{,\eta^2}=\sigma_2'\, \epsilon,\\
\psi^2_{,\eta^1}=\sigma_1'\, \epsilon, \qquad \psi^2_{,\eta^2}=\frac{-2\sigma_1'\, \epsilon\, s_1 + \psi^2\chi_2}{2s_2}.
\end{align*}
The compatibility conditions of these relations lead to equations for $\sigma_1, \sigma_2$:
$$
\frac{{\sigma^1}''}{{\sigma^1}'}+\frac{s_1'}{s_1}+\frac{\chi_1}{2s_1}=0, \quad \frac{{\sigma^2}''}{{\sigma^2}'}+\frac{s_2'}{s_2}+\frac{\chi_2}{2s_2}=0.
$$
Solving these equations we complete the set of Casimirs for all of the above examples.

%\noindent{\bf Remark.} For $n=2$, we have the following general formula for the remaining two `more complicated' Casimirs that works for all cases of Table 3: $$ -u^1 \tau_1 {\int} \frac{\epsilon \, d\eta^1}{s_1 \tau_1 \tau_2} + {\int}\frac{d\eta^2}{s_2\tau_2}, \qquad -u^2 \tau_2 {\int} \frac{\epsilon \, d\eta^2}{s_2 \tau_1 \tau_2} + {\int}\frac{d\eta^1}{s_1\tau_1}, $$ where $\tau_1={ e}^{ {\int}\frac{\chi_1}{2 s_{1}}{d}{\eta^1}}, \ \tau_2={ e}^{ {\int}\frac{\chi_2}{2 s_{2}}{d}{\eta^2}}$.

\section{Hamiltonian formulation of the full hard-rod kinetic equation}
\label{sec:hr}

In this section we construct local Hamiltonian formalism for the full hard-rod kinetic equation. Note that nonlocal Hamiltonian formalism in the general case was discussed in \cite{Bul}. For $G(\mu, \eta)=-a, \ S(\eta)=\eta$,  equation (\ref{gas}) simplifies to 
\begin{equation}\label{hr3}
\begin{array}{c}
f_t+(sf)_x=0,\\
\ \\
{\displaystyle s(\eta)=\eta-a \int f(\mu)[s(\mu)-s(\eta)]\ d\mu}.
\end{array}
\end{equation}

\subsection{Moment representation and linearisation}
\label{sec:mom}

Let us introduce the `moments', 
$$
A^i=\int \eta^i f(\eta)\, d\eta, \quad B^i=\int \eta^i s(\eta) f(\eta)\, d\eta,
$$
$i\in \{0, 1, 2, \dots \}$. Multiplying the first equation (\ref{hr3}) by $\eta^i$ and integrating over $\eta$ one obtains $A^i_t+B^i_x=0$. Similarly, multiplying the second equation (\ref{hr3}) by $\eta^i f(\eta)$ 
and integrating over $\eta$ one obtains the explicit form of $B$-moments in terms of $A$-moments, namely, $B^0=A^1, \ B^k=\frac{A^{k+1}-aA^1A^k}{1-aA^0}, \ k\geq 1$. In particular, the second equation gives $s(\eta)=\frac{\eta-aA^1}{1-aA^0}$, so that the kinetic equation reduces to
\begin{equation}\label{hardrod}
f_t+\left(\frac{\eta-aA^1}{1-aA^0}\, f\right)_x=0;
\end{equation}
in equivalent form, it has appeared in \cite{Perkus}, eq. (2) and \cite{Dobrushin}, eq. (1.1). In terms of the moments, equation (\ref{hardrod}) leads to an infinite hydrodynamic chain,
\begin{equation}\label{chain}
A^i_t+\left(\frac{A^{i+1}-aA^1A^i}{1-aA^0}\right)_x=0,
\end{equation}
$i\in \{0, 1, 2, \dots \}$, which has first appeared  in the classification of integrable Egorov hydrodynamic chains (\cite{P2004},  degeneration of formula (14)). Under delta-functional reduction (\ref{del}), the $A$-moments assume the form $A^i=\sum_{k=1}^n u^k(\eta^k)^i.$
The corresponding  system (\ref{uv}) can be viewed as a $2n$-component hydrodynamic reduction of chain (\ref{chain}); see, e.g., \cite{GT, P2003, FM} for the theory of integrable hydrodynamic chains and their reductions. 

Note that chain (\ref{chain}) linearises under the reciprocal transformation $(t, x) \to (t, y)$ where the new nonlocal independent variable $y$ is defined as
$dy=(1-aA^0)dx+aA^1dt$, followed by the change of moments, $\tilde A^i=\frac{A^i}{1-aA^0}$,  taking equations (\ref{chain}) to
\begin{equation}\label{chainl}
\tilde A^i_t+\tilde A^{i+1}_y=0.
\end{equation}
This allows to obtain a general solution of equation (\ref{hardrod}) as follows (compare with \cite{Perkus}). Applied directly to  equation (\ref{hardrod}), the above reciprocal transformation leads to the linear PDE
$$
\tilde f_t+\eta \tilde f_y=0,
$$
where $\tilde f=\frac{f}{1-aA^0}$. Thus, $\tilde f=\varphi(y-\eta t, \eta)$ where $\varphi$ is some function of the indicated arguments, so that
\begin{equation}\label{fex}
f=(1-aA^0)\varphi(y-\eta t, \eta).
\end{equation}
Integrating  relation (\ref{fex}) with respect to $\eta$ gives $A^0=(1-aA^0)\int \varphi(y-\eta t, \eta)\, d\eta$. Similarly, multiplying (\ref{fex}) by $\eta$ and integrating with respect to $\eta$ gives $A^1=(1-aA^0)\int \eta \, \varphi(y-\eta t, \eta)\, d\eta$. In particular,
$$
\frac{1}{1-aA^0}=1+a\int \varphi(y-\eta t, \eta)\, d\eta, \quad \frac{A^1}{1-aA^0}=\int \eta\, \varphi(y-\eta t, \eta)\, d\eta,
$$
so that the relation $dx=\frac{1}{1-aA^0}dy-\frac{aA^1}{1-aA^0}dt$ can be explicitly integrated for $x$:
$$
x=y+a \int \Phi(y-\eta t, \eta)\, d\eta
$$
where $\Phi$ is the antiderivative of $\varphi$ in the first argument.
Ultimately, the general solution of the kinetic equation for hard-rod gas can be represented in the following parametric form:
$$
f(\eta, y, t)=\frac{\varphi(y-\eta t, \eta)}{1+a\int \varphi(y-\eta t, \eta)\, d\eta},
$$
$$
s(\eta, y, t)=\eta+a\eta \int \varphi(y-\eta t, \eta)\, d\eta-a\int \eta \,  \varphi(y-\eta t, \eta)\, d\eta,
$$
$$
x=y+a \int \Phi(y-\eta t, \eta)\, d\eta,
$$
where $\varphi$ is an arbitrary function of its arguments and $\Phi$ is the antiderivative of $\varphi$ in the first argument.

\subsection{Hamiltonian structures}
\label{sec:Ham12}

Rewriting Hamiltonian structures (\ref{nmhr}) in terms of the moments $A^i$, in the limit $n\to \infty$ one obtains local Hamiltonian structures of the full kinetic equation, represented as hydrodynamic chain (\ref{chain}). As an example, let us take $n_i=1/ (u^i)^2, \ m_i=0$, in which case the metric $g$ (with upper indices) is represented by a symmetric matrix consisting of $n$  blocks of the form $\left(\begin{array}{cc}0&n_i\\ n_i&0\end{array}\right)$. The corresponding symmetric bivector is $\sum \frac{2}{(u^i)^2}\partial_{r^i}\partial_{\eta^i}$. Rewriting it in the variables $A^i$ we obtain a $2n\times 2n$ symmetric matrix which coincides with the upper left $2n\times 2n$ block of the infinite matrix 
\begin{equation}\label{infn}
\frac{J(U+U^T)J^T}{(1-aA^0)}
\end{equation}
where the infinite matrices $J$ and $U$ are defined as
$$
J=\left(\begin{array}{cccc}
(1-aA^0)&0&0&\dots\\
-aA^1&1&0 &\dots\\
-aA^2&0 & 1 &\dots\\
\dots&\dots&\dots&\dots
\end{array}
\right), \qquad U=\left(\begin{array}{cccc}
0&0&0&\dots\\
 A^0& A^1& A^2 &\dots\\
2A^1& 2 A^2 & 2 A^3 &\dots\\
\dots&\dots&\dots&\dots
\end{array}
\right).
$$
Note that $(1-aA^0)J$ is the Jacobian matrix of the change of variables $A\to \tilde A$ discussed in Section \ref{sec:mom}.
Thus, formula (\ref{infn}) gives a flat contravariant metric of  hydrodynamic chain (\ref{chain}). The exact form of the corresponding Hamiltonian operator is 
obtained in Example 1 below.

\medskip
Hamiltonian structures of hydrodynamic chain (\ref{chain}) can also be obtained from that of the linear chain (\ref{chainl}) by utilising reciprocal transformation $dy=(1-aA^0)dx+aA^1dt$ and the change of variables $\tilde A^i=\frac{A^i}{1-aA^0}$; see \cite{F, FP2003} for the behaviour of Hamiltonian structures under reciprocal transformations. We will present two examples of this kind.

\medskip

\noindent{\bf Example 1.} Introducing the column vectors $\tilde A=(\tilde A^0, \tilde A^1, \dots)^T, \ \frac{\partial \tilde h}{\partial \tilde A}=(\frac{\partial \tilde h}{\partial \tilde A^0}, \frac{\partial \tilde h}{\partial \tilde A^1}, \dots)^T$, one can represent the linear chain (\ref{chainl})  in Hamiltonian form,
$$
\tilde {A}_t=\tilde B\, \frac{\partial \tilde h}{\partial \tilde A},
$$
with the  Hamiltonian density $\tilde h=-\frac{1}{2}\tilde A^2$ and the Kupershmidt-Manin Hamiltonian operator 
$$
\tilde B=\tilde U\frac{d}{dy}+\frac{d}{dy}\tilde U^T, \qquad
\tilde U=\left(\begin{array}{cccc}
0&0&0&\dots\\
\tilde A^0&\tilde A^1&\tilde A^2 &\dots\\
2\tilde A^1& 2\tilde A^2 & 2\tilde A^3 &\dots\\
\dots&\dots&\dots&\dots
\end{array}
\right),
$$
which first appeared as a Hamiltonian structure of the Benney chain \cite{KupMan}. Applying the transformations indicated above we obtain a local Hamiltonian formulation of chain (\ref{chain}): 
$$
{A}_t= B\, \frac{\partial  h}{\partial  A},
$$
with the  Hamiltonian density $ h=-\frac{1}{2} A^2$ and the  Hamiltonian operator 
$$
B=J\left(\frac{1}{1-aA^0}U\frac{d}{dx}+\frac{d}{dx}U^T\frac{1}{1-aA^0}\right)J^T
+\frac{a}{1-aA^0}(PA^T_x-A_xP^T).
$$
Here $J$ and $U$ are the same as above and
 the column vector  $P$ is defined as $P=(0, A^0, 2A^1, 3A^2, \dots)^T$. Note that the contravariant metric of this operator coincides with (\ref{infn}).

\medskip

\noindent{\bf Example 2.} One can represent the linear chain (\ref{chainl})  in yet another Hamiltonian form,
$$
\tilde {A}_t=\tilde B\, \frac{\partial \tilde h}{\partial \tilde A},
$$
with the  Hamiltonian density $\tilde h=-\tilde A^1$ and the Kupershmidt Hamiltonian operator \cite{Kup},
$$
\tilde B=\tilde V\frac{d}{dy}+\frac{d}{dy}\tilde V^T, \qquad
\tilde V=\left(\begin{array}{cccc}
\tilde A^0&\tilde A^1&\tilde A^2&\dots\\
\tilde A^1&\tilde A^2&\tilde A^3 &\dots\\
\tilde A^2& \tilde A^3 & \tilde A^4 &\dots\\
\dots&\dots&\dots&\dots
\end{array}
\right).
$$
 Applying the transformations indicated above we obtain a nonlocal (constant curvature) Hamiltonian formulation of chain (\ref{chain}): 
$$
{A}_t= B\, \frac{\partial  h}{\partial  A},
$$
with the  Hamiltonian density $ h=-A^1$ and the nonlocal Hamiltonian operator 
$$
B=J\left(\frac{1}{1-aA^0}V\frac{d}{dx}+\frac{d}{dx}V^T\frac{1}{1-aA^0}\right)J^T
+2a\, (AA^T_x-A_xA^T)+2a\, A_x\left(\frac{d}{dx}\right)^{-1}A^T_x.
$$
Here $A$ and  $J$ denote the same as in Example 1 while
$$
V=\left(\begin{array}{cccc}
A^0&A^1&A^2&\dots\\
 A^1& A^2& A^3 &\dots\\
A^2&  A^3 & A^4 &\dots\\
\dots&\dots&\dots&\dots
\end{array}
\right).
$$

\subsection{Monge-Amp\`ere form of delta-functional reductions}

Here we establish a link of delta-functional reductions of the hard-rod kinetic equation to the higher-order Monge-Amp\`ere equations discussed in \cite{Fer02}. 
As we saw in Section \ref{sec:mom}, in terms of the moments $A^i$ the delta-functional reduction assumes the form $A^i=\sum_{k=1}^n u^k(\eta^k)^i.$ For definiteness, let us consider the case $n=2$ (the general case is similar). Explicitly, we have
\begin{align}\label{AA}\begin{split}
&A^0=u^1+u^2, \\
&A^1=u^1\eta^1+u^2\eta^2, \\
&A^2=u^1(\eta^1)^2+u^2(\eta^2)^2, \\
&A^3=u^1(\eta^1)^3+u^2(\eta^2)^3, \\
&A^4=u^1(\eta^1)^4+u^2(\eta^2)^4,
\end{split}\end{align}
equivalently, 
$$
A=\left(
\begin{array}{ccc}
A^0&A^1&A^2\\
 A^1& A^2& A^3 \\
A^2&  A^3 & A^4\\
\end{array}
\right) =u^1
\left(
\begin{array}{ccc}
1 & \eta^1 & (\eta^1)^2 \\
\eta^1 & (\eta^1)^2 & (\eta^1)^3 \\
(\eta^1)^2 & (\eta^1)^3 & (\eta^1)^4 
\end{array}
\right)+
u^2
\left(
\begin{array}{ccc}
1 & \eta^2 & (\eta^2)^2 \\
\eta^2 & (\eta^2)^2 & (\eta^2)^3 \\
(\eta^2)^2 & (\eta^2)^3 & (\eta^2)^4 
\end{array}
\right).
$$
As both matrices on the right-hand side have rank one, the Hankel matrix $A$ on the left has rank two, so that $\det A=0$,  providing explicit expression for $A^4$ in terms of $A^0, A^1, A^2, A^3$. Thus, chain (\ref{chain}) truncates to the first four equations:
\begin{equation*}
\begin{array}{c}
A^0_t+A^1_x=0, \quad A^1_t+\left(\frac{A^{2}-a(A^1)^2}{1-aA^0}\right)_x=0,\quad A^2_t+\left(\frac{A^{3}-aA^1A^2}{1-aA^0}\right)_x=0, \quad A^3_t+\left(\frac{A^{4}-aA^1A^3}{1-aA^0}\right)_x=0.
\end{array}
\end{equation*}
Under the reciprocal transformation $dy=(1-aA^0)dx+aA^1dt$, followed by the change of variables $\tilde A^i=\frac{A^i}{1-aA^0}$, this system simplifies to 
\begin{equation}\label{tA}
\begin{array}{c}
\tilde A^0_t+\tilde A^1_y=0, \quad \tilde A^1_t+\tilde A^2_y=0, \quad \tilde A^2_t+\tilde A^3_y=0, \quad \tilde A^3_t+\tilde A^4_y=0,
\end{array}
\end{equation}
where $\tilde A^4$ is expressed in terms of $\tilde A^0, \tilde A^1, \tilde A^2, \tilde A^3$ from the equation $\det \tilde A=0$ ($\tilde A$ is the same matrix as $A$ but with tilded entries). Finally, introducing a potential $f(t, y)$ such that
$$
\tilde A^0=f_{yyyy}, \ \tilde A^1=-f_{yyyt}, \ \tilde A^2=f_{yytt}, \ \tilde A^3=-f_{yttt}, \ \tilde A^4=f_{tttt}, 
$$
we reduce system (\ref{tA}) to a single fourth-order Monge-Amp\`ere equation for $f$,
$$
\det \tilde A=\det 
\left(
\begin{array}{ccc}
f_{yyyy} & -f_{yyyt} & f_{yytt} \\
-f_{yyyt} & f_{yytt} & -f_{yttt} \\
f_{yytt} & -f_{yttt} & f_{tttt} 
\end{array}
\right)=0.
$$
Note that all minuses can be removed by the change of sign of $y$; we refer to \cite{Fer02} for the properties of this equation and its generalisations. 

Let us mention that equation (\ref{AA}) has a simple geometric meaning: consider  projective space $\mathbb{P}^4$ with affine coordinates $A^0, \dots, A^4$, supplied with a rational normal curve $\gamma=(1, s, s^2, s^3, s^4)$. Then equations (\ref{AA}) parametrise bisecant variety of $\gamma$ (collection of all bisecant lines  of $\gamma$), which is a cubic hypersurface in $\mathbb{P}^4$ with the equation $\det A=0$.

\section{Concluding remarks}

Table 3 suggests that, in the limit $n\to \infty$, the local Hamiltonian structures of Table 3 should go to local Hamiltonian structures of the corresponding full kinetic equations. At the level of Hamiltonian densities and momenta,  this limit is easy to see. For example, the Hamiltonian density of delta-functional reduction (\ref{del}) of the KdV kinetic equation is given by
$h= -\frac{4}{3} \sum_{i=1}^n u^i(\eta^i)^2$,
which can be represented as 
$$
h=-\frac{4}{3}\int \eta^2f(\eta, x, t)\, d\eta=-\frac{4}{3}A^2
$$
where $A^2$ is one of the moments introduced in Section \ref{sec:hr}. Thus, as a Hamiltonian of the full KdV kinetic equation one should take
\begin{equation*}\label{Hfull}
H=-\frac{4}{3}\int \int \eta^2f(\eta, x, t)\, d\eta dx=-\frac{4}{3}\int A^2 dx.
\end{equation*}
Similarly, the density of momentum  is given by $g= \sum_{i=1}^n u^i$,
which can be represented as
\begin{equation*}\label{hfull}
g=\int f(\eta, x, t)\, d\eta=A^0.
\end{equation*}
Thus, as a momentum of the full KdV kinetic equation one should take
\begin{equation*}\label{Hfull}
G=\int \int f(\eta, x, t)\, d\eta dx=\int A^0 dx.
\end{equation*}
Unfortunately, it is not entirely clear how to write the metric of the Hamiltonian operator in terms of the moments $A^i$, as well as how to pass to the limit as $n\to \infty$.

\section{Acknowledgements}
We thank G. El and M. Pavlov for useful discussions.
PV's research was partially supported by GNFM of the Istituto Nazionale di Alta Matematica
(INdAM), the research project Mathematical Methods in Non-Linear Physics
(MMNLP)  by the Commissione Scientifica Nazionale -- Gruppo 4 -- Fisica Teorica
of the Istituto Nazionale di Fisica Nucleare (INFN) and by "Borse per viaggi all'estero" 
of the Istituto Nazionale di Alta Matematica, which permitted to visit the Geometry and 
Mathematical Physics group at Loughborough University.
%The research of EVF was supported by a grant from the Russian Science Foundation  No. 21-11-00006, https://rscf.ru/project/21-11-00006/.

\end{document}